\documentclass[11pt]{amsart}
\usepackage{etoolbox}
\usepackage[margin=1in]{geometry}

\usepackage{parskip}
\usepackage{amssymb, amsmath, amsthm}
\usepackage{breqn}
\usepackage{verbatim}
\usepackage{enumerate}
\usepackage{fancyhdr}
\usepackage{graphicx}
\usepackage{varwidth}
\usepackage{bm}
\usepackage{authblk}
\usepackage[thinlines]{easytable}

\usepackage{hyperref}

\usepackage{algpseudocode}
\usepackage{float}
\floatstyle{boxed}
\restylefloat{figure}
\restylefloat{table}

\usepackage[symbol]{footmisc}

\theoremstyle{plain}
\newtheorem{mythm}{Theorem}[section]

\newtheorem{myprop}[mythm]{Proposition}

\theoremstyle{definition}

\newcommand{\mc}[1]{\ensuremath{\mathcal{#1}}}

\newcommand{\mb}[1]{\ensuremath{\mathbb{#1}}}
\newcommand{\ra}{\rightarrow}
\newcommand{\xra}{\xrightarrow}

\newcommand{\N}{\mb{N}}

\newcommand{\F}{\mb{F}}

\newcommand{\Con}{\ensuremath{{\rm Con}}}
\newcommand{\ev}{\ensuremath{{\rm ev}}}

\usepackage[foot]{amsaddr}

\begin{document}

\title{Subquadratic time encodable codes beating the Gilbert-Varshamov bound}

\author{Anand Kumar Narayanan\textsuperscript{1}}
\address{\textsuperscript{1} Laboratoire d'Informatique de Paris 6, Pierre et Marie Curie University, Paris.}
\email{anand.narayanan@lip6.fr}
\author{Matthew Weidner\textsuperscript{2}}
\address{\textsuperscript{2} Department of Mathematics, California Institute of Technology, Pasadena.}
\email{mweidner@caltech.edu}
\setcounter{Maxaffil}{0}
\renewcommand\Affilfont{\itshape\small}
\maketitle
\vspace{-10pt}

\footnotetext{\textsuperscript{1} Supported by NSF grant \#CCF-1423544, Chris Umans' Simons Foundation Investigator grant and European Union's H2020 Programme under grant agreement number ERC-669891.}
\footnotetext{\textsuperscript{2} Supported by NSF grant \#CCF-1423544, Chris Umans' Simons Foundation Investigator grant, and the Rita A.\ and \O istein Skjellum SURF Fellowship.}

\begin{abstract}
We construct explicit algebraic geometry codes built from the Garcia-Stichtenoth function field tower beating the Gilbert-Varshamov bound for alphabet sizes at least $19^2$.  Messages are identified with functions in certain Riemann-Roch spaces associated with divisors supported on multiple places. Encoding amounts to evaluating these functions at degree one places. By exploiting algebraic structures particular to the Garcia-Stichtenoth tower, we devise an intricate deterministic $\omega/2 <1.19$ runtime exponent encoding and $1+\omega/2 < 2.19$ expected runtime exponent randomized (unique and list) decoding algorithms. Here $\omega < 2.373$ is the matrix multiplication exponent. If $\omega=2$, as widely believed,  the encoding and decoding runtimes are respectively nearly linear and nearly quadratic. Prior to this work, encoding (resp. decoding) time of code families beating the Gilbert-Varshamov bound were quadratic (resp. cubic) or worse.


\end{abstract}

\thispagestyle{empty}
\pagebreak

\setcounter{page}{1}

\section{Introduction}
\subsection{Codes Beating the Gilbert-Varshamov Bound}
Error-correcting codes enable reliable transmission of information over an erroneous channel. A (block) error-correcting code of block length $N$ over a finite alphabet $\Sigma$ of size $Q$ is a subset $\mathcal{C} \subseteq \Sigma^N$. The rate $R$ at which information is transmitted through the code $\mathcal{C}$ is defined as $\log_Q(|\mathcal{C}|)/N$. The minimum distance $d$ of the code $\mathcal{C}$, defined as the minimum Hamming distance among all distinct pairs of elements (codewords) in the code $\mathcal{C}$, quantifies the number of errors that can be tolerated. A code with minimum distance $d$ can tolerate $(d-1)/2$ errors. The relative distance $\delta$ is defined as $\delta:=d/N$. A code $\mathcal{C}$ is linear if the alphabet is a finite field $\F_Q$ (with $Q$ elements) and $\mathcal{C}$ is an $\F_Q$-linear subspace of $\F_Q^N$. 

\noindent One typically desires codes to transmit information at a high rate while still being able to correct a large fraction of errors. That is, one wants codes with large rate and large relative distance. However, rate and relative distance are competing quantities with a tradeoff between them.  The Gilbert-Varshamov bound assures, for every $Q$, $R>0$, $0< \delta \leq 1-1/Q$ and small positive $\epsilon$, the existence of an infinite family of codes with increasing block length over an alphabet of size $Q$ with rate $R$ and relative distance $\delta$ bounded by
\begin{equation}\label{equation_gv_bound}
R + H_{Q}(\delta) \geq 1-\epsilon,
\end{equation} 
where $H_Q$ is the $Q$-ary entropy function \cite{gil,var}. 
Random linear codes, where one chooses a random subspace of $N$-tuples over a finite field meet the bound with high probability. In fact, Varshamov proved the bound using the probabilistic method with random linear codes. Testing if a given linear code meets the Gilbert-Varshamov bound comes down to approximating the minimum distance, an intractable  task unless NP equals RP \cite{DMS,var1}. Hence constructing codes meeting or beating the Gilbert-Varshamov bound remained a long-standing open problem, until the advent of algebraic geometry codes.

\noindent Goppa proposed algebraic geometry codes obtained from curves over finite fields as a generalization of Reed-Solomon codes \cite{gop}. Messages are identified with functions on the curve in the Riemann-Roch space corresponding to a chosen divisor with support disjoint from a large set of $\F_Q$-rational points on the curve. Evaluations of functions in the Riemann-Roch space at these $\F_Q$-rational points on the curve is taken as the code. The rate of the code is the ratio of the dimension of the Riemann-Roch space to the number of $\F_Q$-rational points. The Riemann-Roch theorem gives a bound on the dimension of the code and yields the following tradeoff between rate and relative distance:
\begin{equation}
	R+\delta \geq 1-\frac{g}{N},
\end{equation}
where $g$ denotes the genus of the curve. This spurred an effort to construct curves over finite fields where the fraction $g/N$ of the genus to the number of $\F_Q$-rational points is as low as possible. Researchers had to contend with the lower bound $g/N \geq 1/(\sqrt{Q}-1)$ of Drinfeld-Vl\u{a}du\c{t} \cite{dv}. In seminal papers, Ihara \cite{ihara} and Tsfasman, Vl\u{a}du\c{t}, and Zink \cite{TVZ} constructed curves meeting the Drinfeld-Vl\u{a}du\c{t} bound when the underlying finite field size $Q$ is a square, leading to the Tsfasman-Vl\u{a}du\c{t}-Zink bound
\begin{equation}\label{equation_tvz}
R+\delta \geq 1-\frac{1}{\sqrt{Q}-1}.
\end{equation}
Remarkably, for $Q\geq 7^2$, the Tsfasman-Vl\u{a}du\c{t}-Zink bound is better than the Gilbert-Varshamov bound! This is a rare occasion where an explicit construction yields better parameters than guaranteed by randomized arguments. Garcia and Stichtenoth described an explicit tower of function fields which meet the Drinfeld-Vl\u{a}du\c{t} bound, hence yield codes matching the Tsfasman-Vl\u{a}du\c{t}-Zink bound \cite{GS}. The curves in the Garcia-Stichtenoth tower are the primary objects of study in our paper. An outstanding open problem in this area is to explicitly construct codes meeting or beating the Gilbert-Varshamov bound over small alphabets, in particular binary codes ($Q=2$).  It is known that algebraic geometry codes beat the Gilbert-Varshamov bound over $\F_Q$ for any prime power $Q \ge 49$ which is not prime and not 125 \cite{bassa}.
\subsection{Linear Time Encodable Codes Meeting the Gilbert-Varshamov Bound}
For codes to find use in practice, one often requires fast encoding and decoding algorithms in addition to satisfying a good tradeoff between rate and minimum distance. 
An encoding algorithm maps a given message to a codeword. A decoding algorithm takes a possibly corrupted codeword, called the received word, and outputs the message that induced it, provided the number of errors is within the designed tolerance.

\noindent A natural question, which remains unresolved, is if there exist linear time encodable and decodable codes meeting or beating the Gilbert-Varshamov bound. One cannot look to random linear codes to resolve this problem, for they require quadratic runtime to encode and are NP-hard to decode \cite{bmt}. In a breakthrough, Speilman, using explicit expander codes, proved the existence of linear time encodable and decodable ``good'' codes \cite{Spe}. A family of codes with increasing block length is deemed good if (in the limit) the rate and relative distance are simultaneously bounded away from zero. A code being good is a weaker condition than meeting the Gilbert-Varshamov bound. Guruswami and Indyk constructed linear time encodable and decodable expander codes approaching the Gilbert-Varshamov bound \cite{GI}. However, the closer one wishes to approach the Gilbert-Varshamov bound, the larger the alphabet size of the code. Druk and Ishai constructed linear time encodable codes meeting the Gilbert-Varshamov bound, but these codes are likely NP-hard to decode \cite{DI}.

\subsection{Main Results}
Our main result is the explicit algebraic construction of subquadratic time encodable codes beating the Gilbert-Varshamov bound, along with an efficient decoding algorithm. 

\begin{mythm}\label{theorem_main_2}
	For every square prime power $Q$ and rate $R \in  \left(0,1- \frac{2\sqrt{Q}+1}{Q-\sqrt{Q}}\right)$,
	 there exists an infinite sequence of codes over $\F_Q$ of increasing length $N$ with rate $R$ and relative distance $\delta$ satisfying  $$R + \delta \geq 1- \frac{2\sqrt{Q}+1}{Q-\sqrt{Q}}.$$
	Further, there exists deterministic algorithms to
	\begin{itemize}
		\item pre-compute a representation of the code at the encoder and decoder in $O(N^{3/2}\log^3 N )$ time; this representation occupies $O(N)$ space
		\item encode a message in $O(N^{\omega/2})$ time, where $\omega$ is the matrix multiplication exponent
	\end{itemize}
    and Las Vegas randomized algorithms to
    \begin{itemize}
	\item decode close to half the designed distance in $O(N^{1+\omega/2}\log^2 N)$ expected time	
	\item list decode up to $N\left(1-\sqrt{2\left(R+2\frac{2\sqrt{Q}+1}{Q-\sqrt{Q}}\right)}-\frac{2\sqrt{Q}+1}{Q-\sqrt{Q}}\right)$ errors with list size at most \\ $\sqrt{\frac{2(\sqrt{Q}-1)}{2 + R(\sqrt{Q}-1)}}$ in $O(N^{1+\omega/2}\log^2 N)$ expected time (requiring additional pre-processing taking $O(N^{\omega})$ time and $O(N^2)$ space).
    \end{itemize}

\end{mythm}
For $Q \geq 19^2$, the tradeoff assured by Theorem \ref{theorem_main_2} is better than the Gilbert-Varshamov bound. The encoding time is linear if the matrix multiplication exponent $\omega$ is indeed $2$ as widely conjectured. The best known bound for $\omega$ yields an encoding time exponent of $1.19$ \cite{lG}.

\noindent The pre-processing step can be thought of as computing a succinct representation of the code and is performed at the encoder and decoder independently. If one desires, the pre-processing and encoding time can be made to approach linear time at the cost of needing larger alphabets to beat the Gilbert-Varshamov bound. Our construction is parametrized by an integer $k \geq 2$ and yields the tradeoff $$R + \delta \geq 1- \frac{k\sqrt{Q}+k-1}{Q-\sqrt{Q}}$$ with pre-computation requiring time $O(N^{3/k}\log^3 N)$, the resulting succinct representation requiring $O(N^{2/k})$ space, and encoding requiring time $O(N^{1 + (\omega - 2)/k})$; see Theorem \ref{thm_main_encoding}.  For decoding, we get pre-computation requiring time $O(N^\omega)$, the resulting succinct representation requiring $O(N^2)$ space, and decoding requiring expected time $O(N^{2 + (\omega-2)/k}\log^2 N)$; see Theorem \ref{thm_main_decoding}.  Theorem \ref{theorem_main_2} corresponds to $k=2$. Table \ref{k_tradeoff} gives a comparison of encoding for $k = 2, 3$.  The likeness to the Tsfasman-Vl\u{a}du\c{t}-Zink bound (equation \ref{equation_tvz}), which is obtained if one is allowed to substitute $k=1$, is striking.

\subsection{Applications}
\noindent As with other algebraic geometry codes, our codes may be used as outer codes in concatenation to obtain long binary codes. In particular, concatenating with Walsh-Hadamard inner codes yields balanced binary codes (or equivalently small bias spaces) \cite{bt}[\S~3.2]. Outside coding theory, our codes have several complexity theoretic implications: efficient secret sharing schemes and interactive proof protocols to name a few. Chen and Cramer \cite{cc} initiated the use of algebraic geometry codes in  secret sharing  and in multi-party computation. Their scheme retains the salient features of Shamir's \cite{sha} secret sharing (to the extent possible) yet only requires small alphabets for the secret shares. Our codes fit seamlessly in their framework resulting in a significant speed up. The runtime exponent (in the number of players) with our codes is $\omega/2$ for secret sharing and $1+\omega/2$ for secret recovery. If $w=2$, we get nearly linear time secret sharing and nearly quadratic time secret recovery. The Chen-Cramer scheme has spawned several extensions and improvements in the ensuing decade, including multiplicative ramp secret sharing schemes \cite{cchp} (for communication efficient secure multi-party computation), high information rate ramp schemes \cite{ccghv}[\S~4.3] (for threshold secure computation) and secret sharing schemes relying on nested algebraic geometry codes (\cite{gmmmr}[Thm 22] and \cite{mar}[Prop 21]). Our codes are applicable across these schemes and result in a speedup similar to that for the Chen-Cramer scheme.

\noindent Exciting recent developments in interactive proofs are promising grounds for applying our codes. In Delegated Computation, a remote server runs a computation for a client and tries to prove interactively that it indeed correctly performed the computation. This scenario was modelled \cite{gkr,rrr} as an interactive proof system where the honest prover (server) is limited to polynomial time computation and the verifier (client) is limited to nearly linear time computation. In a recent breakthrough \cite{rrr}, constant round protocols under this framework were described using \textit{Probabilistically Checkable Interactive Proofs} (PCIPs), an interactive version of PCPs where the verifier only reads a few bits of the transcripts. Independently \cite{bcs}, to improve the efficiency of PCPs by adding rounds of interaction (in the random oracle model), \textit{Interactive Oracle Protocols} (IOPs) were introduced, which are equivalent to PCIPs. Constant rate and constant query IOPs were recently constructed using tensor products of algebraic geometry codes from Garcia-Stichtenoth towers \cite{bcgrs}[\S~5.2, Thm 7.1, Lem 7.2]. Taking tensor products of our codes improves the efficiency of these IOPs. In particular, we reduce the efficiency exponent $c$ from $c>3$ to $c>3/2$ in \cite{bcgrs}[Lem 7.2] by constructing asymptotically good systematic subcodes of our codes (see \S~\ref{subsection_iop}). To this end, we tailored these systematic subcodes in a manner that they can be encoded and (equivalently) checked (that is, decide if a given word is a codeword) with deterministic runtime exponent $3/2$. We anticipate further fruitful applications of our code to these interactive protocols in the future.

\subsection{Code Construction: Riemann-Roch spaces from Shifting}

\noindent We next recount algebraic geometry codes before sketching our construction.  Let $\mathcal{X}$ be a smooth projective (not necessarily plane) curve over a finite field $\F_Q$ and $\F_Q(\mathcal{X})$ the associated function field. A set $\mathcal{P}$ of $\F_Q$-rational points on $\mathcal{X}$ (or equivalently, degree $1$ places in $\F_Q(\mathcal{X})$) will serve as the code places. A divisor is chosen, typically of the form $r P_\infty$ where $P_\infty$ is a place in $\F_Q(\mathcal{X})$ away from $\mathcal{P}$ and $r \in \N$. The Riemann-Roch space $\mc{L}(rP_\infty)$ (consisting of functions in $\F_Q(\mathcal{X})$ whose poles are confined to $P_\infty$ and have order bounded by $r$) is identified with the message space. The code is the evaluation of $\mc{L}(rP_\infty)$ at the places in $\mathcal{P}$. The rate is determined by the dimension of $\mc{L}(rP_\infty)$ as an $\F_Q$-linear space. The Riemann-Roch theorem then yields bounds on the rate and relative distance, thereby quantifying the performance of the code. Encoding messages requires efficient algorithms to construct and evaluate functions from the Riemann-Roch space. This can be accomplished in polynomial time due to algorithms of Huang and Ierardi (for smooth projective plane curves) \cite{HI} and Hess (for smooth projective curves) \cite{Hess, GS1}. However, such generic algorithms are far from linear.  We focus on building fast algorithms tailored to the Garcia-Stichtenoth tower.

Take $Q=q^2$ where $q$ is a prime power. The Garcia-Stichtenoth tower over $\F_{q^2}$ is the sequence of function fields defined by $F_0 = \F_{q^2}(x_0)$, and $F_{i+1} = F_i(x_{i+1})$ where $x_{i+1}$ satisfies the relation
\[
x_{i+1}^q + x_{i+1} = \frac{x_i^q}{x_i^{q-1} + 1}.
\]
Let $P^{(n)}_\infty$ denote the unique pole of $x_0$ in $F_n$.  In a series of works, Aleshnikov, Deolalikar, Kumar, Shum and Stichtenoth \cite{AKSS,ADKSS,S} described the splitting of places in $F_n$ and established a pole-cancelling algorithm to compute a basis for the Riemann-Roch spaces $\mc{L}(rP^{(n)}_\infty)$. This culminated in a quadratic time algorithm (with nearly cubic time pre-processing) to encode these codes. 


\noindent From the $n^{th}$ function field $F_n$, we construct codes of block length $q^n(q^2-q)$. Let a small integer parameter $k \ge 2$ be chosen.  Assume for ease of exposition that $k$ divides $n$.

\textbf{Code places:}  Let $\Omega = \{\alpha \in \F_{q^2} \mid \alpha^q + \alpha = 0\}$, which has size $q$. The code places are all the places in $F_n$ which are zeros of $x_0 - \alpha$ for $\alpha \in \F_{q^2} \setminus \Omega$; there are $q^n(q^2 - q)$ such places, all of which are $\F_{q^2}$-rational.

\textbf{Message Space:} We begin by constructing functions in the lower function field $F_{n/k}$ that are regular. By regular, we mean that their poles in $F_n$ are confined to $P^{(n)}_\infty$.  We devise a procedure called \textit{shifting} that translates functions in $F_{n/k}$ to $F_n$. It is quite simple and, in spirit, just relabelling the subscripts so that each $x_j$ becomes $x_{j+i}$ for some chosen positive integer $i$. The symmetry of the defining equations of the Garcia-Stichtenoth tower allows us to determine the pole divisor in $F_n$ of shifts of regular functions from $F_{n/k}$. The poles of the shifts are not confined to $P^{(n)}_\infty$. However, by taking products of carefully chosen shifted functions, we can bound the new poles arising outside $P^{(n)}_\infty$. We thus construct mostly regular functions in $F_n$ by taking products of shifts of regular functions in $F_{n/k}$.

In summary, given a positive integer $r$, we can construct a large number of functions in a Riemann-Roch space of the form $\mc{L}(G+rP^{(n)}_\infty)$ where $G$ is a small pole divisor. We take the span of these constructed functions to be our message space. Enough functions are constructed to yield large rate codes, and properties of curves applied to the divisor $G+rP^{(n)}_\infty$ yield a lower bound on the minimum distance. 

\textbf{Pre-computation:} To aid in rapid encoding and decoding, we first pre-compute a set of regular functions in $F_{n/k}$. In particular, we pre-compute the evaluations at all code places of a basis for all regular functions in $F_{n/k}$ of a certain bounded pole degree. This Riemann-Roch space computation is performed using an algorithm of Shum et.\ al.\ \cite{ADKSS} and has runtime exponent $3/k$.

\textbf{Encoding Algorithm:}
Given such a set of regular functions in $F_{n/k}$, we compute basis functions in our message space by taking products of shifted functions. Given a message, which is a tuple over $\F_{q^2}$, encoding amounts to evaluating the corresponding linear combination of the basis functions simultaneously at the code places. We devise a Baby-Step Giant-Step algorithm to perform this multipoint evaluation. The runtime of the encoding step depends on the parameter $k$. For $k=2$, the crux of the computation is square matrix multiplications, resulting in an encoding algorithm with runtime exponent $\omega/2$. For larger values of $k$, the crux is rectangular matrix multiplications of shape determined by $k$, and the runtime is again subquadratic. In particular, larger values of $k$ give rise to faster encoding algorithms. 

\textbf{Decoding Algorithms:} We tailor the Shokrollahi-Wasserman algorithm \cite{SW} to our code setting to uniquely decode close to half the relative distance and list decode beyond that. The Shokrollahi-Wasserman algorithm first interpolates a polynomial with coefficients in a Riemann-Roch space such that each message sufficiently close to the received word is a root. Then the roots of the interpolated polynomial in the message Riemann-Roch space are enumerated. Finally the encoding algorithm is used to verify and output the messages in the enumeration that are indeed sufficiently close to the received word.

To adapt their algorithm to our setting, we identify an appropriate Riemann-Roch space (which incidentally is an extension of the message space) as the coefficient space of the interpolation polynomial. Determining the interpolation polynomial now boils down to solving a linear system. We observe that computing matrix-vector products corresponding to this linear system is virtually identical to encoding messages, a task accomplished in subquadratic time. Invoking Wiedemann's algorithm \cite{Wie} (a Las Vegas randomized iterative method involving matrix-vector products) to solve the linear system, we obtain the interpolation polynomial in subcubic expected time. 

To unique decode, we restrict the interpolation polynomial to have degree one. The root finding step is trivial and the algorithm corrects errors up to nearly half the designed distance.

To correct beyond half the designed distance we allow interpolation polynomials of degree greater than one. We perform root finding in quadratic time provided an extra pre-processing step requiring quadratic storage. The resulting algorithm corrects $N \left(1-\sqrt{2\left(R+2\frac{kq+k-1}{q(q-1)}\right)}-2\frac{kq+k-1}{q(q-1)} \right)$ errors with list size at most $\sqrt{\frac{2(q-1)}{2 + R(q-1)}}$.  For $k=2$, with list size at most $2$, it corrects at least as many errors as guaranteed in the trade off in Theorem \ref{theorem_main_2}.


\subsection{Organization:} In \S \ref{section_splitting}, we recount results from 
\cite{AKSS} on the splitting of places in the Garcia-Stichtenoth tower and establish notation. In \S \ref{section_mostlyregular}, we define the shifting operation and construct mostly regular functions in $F_n$ from regular functions in $F_{n/k}$. The code sequences derived from mostly regular functions are defined and their parameters established in \S \ref{section_codeconstruction}. In \S \ref{section_encoding} we develop the subquadratic time encoding algorithm using fast matrix multiplication. The decoding algorithms are presented in \S \ref{section_decoding}.

\section{Splitting of Places in the Garcia-Stichtenoth Tower}\label{section_splitting}
In this section, we recall some notation and results from \cite{AKSS} on the splitting of places in the Garcia-Stichtenoth tower.

In $F_0 = \F_{q^2}(x_0)$, $x_0$ has a unique pole, which we denote by $P^{(0)}_\infty$.  This place is totally ramified in every field extension $F_n/F_0$, hence there is a unique place lying above $P^{(0)}_\infty$ in $F_n$; we denote this place by $P^{(n)}_\infty$. Let $\Omega = \{\alpha \in \F_{q^2} \mid \alpha^q + \alpha = 0\}$, which has size $q$.  For $\alpha \in \F_{q^2}$, let $P^{(0)}_\alpha$ denote the unique zero of $x_0 - \alpha$ in $F_0$.  When $\alpha \in \F_{q^2} \setminus \Omega$, $P^{(0)}_\alpha$ splits completely in every field extension $F_n/F_0$, yielding $q^n$ $\F_{q^2}$-rational places.  As $\alpha$ varies, we get $(q^2 - q)q^n$ rational places in $F_n$, which we take to be the set of code places.

When $\alpha \in \Omega \setminus \{0\}$, $P^{(0)}_\alpha$ is totally ramified in every field extension $F_n/F_0$, hence there is a unique place lying above it in $F_n$; we denote this place by $P^{(n)}_\alpha$. The most interesting place is $P^{(0)}_0$.  For $t \ge 1$, let $S^{(t-1)}_t$ denote the unique place in $F_{t-1}$ that is a zero of $x_{t-1}$.  We sometimes treat $S^{(t-1)}_t$ as a singleton set instead of a place.  We have $S^{(0)}_1 = P^{(0)}_0$, and $S^{(u-1)}_u$ lies over $S^{(t-1)}_t$ whenever $u \ge t$.

In the field extension $F_t/F_{t-1}$, $S^{(t-1)}_t$ splits completely.  Specifically, for each $\alpha \in \Omega$, there is a unique place of $F_t$ which is a simultaneous zero of $x_{t-1}$ and $x_t - \alpha$, and these are all of the places lying above $S^{(t-1)}_t$.  We let $S^{(t)}_t$ denote the set of all places lying above $S^{(t-1)}_t$ in $F_t$ besides $S^{(t)}_{t+1}$, i.e., $S^{(t)}_t$ contains the simultaneous zero of $x_{t-1}$ and $x_t - \alpha$ for each $\alpha \in \Omega \setminus \{0\}$.  For $u \ge t$, we let $S^{(u)}_t$ denote the set of all places of $F_u$ lying above a place in $S^{(t)}_t$.  For $t \le t' \le u + 1$, we let $S^{(u)}_{t, t'} := \bigcup_{i=t}^{t'} S^{(u)}_i$.

Let $u \ge t - 1$.  If $u < 2t$, then all of the places in $S^{(u)}_t$ are unramified (but not necessarily split) in the field extension $F_{u+1}/F_u$.  If $u \ge 2t$, then all of the places in $S^{(u)}_t$ are totally ramified in $F_{u+1}/F_u$.

If $Q$ is a place of $F_n$ and $x \in F_n$, we let $v_Q(x)$ denote the valuation of $x$ at $Q$.  We define the weight of $x \in F_n$ to be $-v_{P^{(n)}_\infty}(x)$.

\section{Mostly Regular Functions Through Shifting}\label{section_mostlyregular}
For the remainder of the paper, fix an integer parameter $k \ge 2$.  In this section, given a positive integer $r$, through shifting we construct a mostly regular function $f_r\in F_n$ of weight precisely $r + \sum_{i=1}^{k} q^{n - (i-1)\lceil n/k \rceil + 1}$ from regular functions in $F_{n/k}$. The discrepancy of $f_r$ from being regular will be quantified by a pole divisor $G$ of degree at most $q^n(kq + k - 1)$. That is, there is a pole divisor $G$ of said degree and weight $\sum_{i=1}^{k} q^{n - (i-1)\lceil n/k \rceil + 1}$ such that for all $r$, $$f_r  \in \mc{L}(G + r(P^{(n)}_\infty)) \setminus \mc{L}(G + (r-1)(P^{(n)}_\infty)).$$
Once a choice of a regular function of each weight in $F_{n/k}$ used by our construction is fixed, the functions $f_r$ are uniquely determined.   
\subsection{Shifting}
First, we define the shifting operation and determine the poles in $F_n$ of functions arising out of shifting regular functions in $F_{n/k}$.  

Let $f = x_m^{e_m} \cdots x_0^{e_0}$ be a monomial in $F_m$.  For $i \ge 0$, we define the \textit{shift of $f$ by $i$} to be the element
\[
f[i] := x_{m+i}^{e_m} \cdots x_i^{e_0} \in F_{m+i}.
\]
We extend the definition of shift $\F_{q^2}$-linearly to all of $F_m$.

We will use the following notations throughout the paper.  For a function field $E$ and an element $x \in E$, let $(x)^E$ and $(x)^E_\infty$ denote the principal divisor and pole divisor of $x$ as an element of $E$.  In the case $E = F_n$ for some $n$, we use the abbreviations $(x)^{(n)}$ and $(x)^{(n)}_\infty$ in place of $(x)^{F_n}$ and $(x)^{F_n}_\infty$.  For a finite extension of function fields $E'/E$, let $\Con^{E'}_{E}$ denote the corresponding conorm map; this is the unique homomorphism from the divisor group of $E$ to the divisor group of $E'$ such that for all places $Q$ of $E$,
\[
\Con^{E'}_{E}(Q) = \sum_{Q' | Q} e(Q' | Q) \cdot Q',
\] 
where the sum runs over all places $Q'$ of $F_{E'}$ lying over $Q$ and where $e(Q' | Q)$ denotes the ramification index of $Q'$ over $Q$.  We have the identities
\begin{align*}
(x)^{E'} &= \Con^{E'}_{E}\left((x)^E\right) \\
(x)^{E'}_\infty &= \Con^{E'}_{E}\left((x)^E_\infty\right)
\end{align*}
for all $x \in E$.  Also, for all divisors $D$ of $E$, we have $\deg \Con^{E'}_E(D) = [E': E]\deg(D)$.

\begin{myprop}\label{prop_shift}
	\begin{enumerate}[(a)]
		\item For any $f \in F_m$ and $i \ge 0$, $f[i]$ is well-defined, i.e., it does not depend on the representation of $f$ as a sum of monomials.
		\item Let $f \in F_m$ be regular of weight $r$.  Then:
		\begin{itemize}
			\item $f[i]$ has weight $r$.
			\item $f[i]$ is regular at $S^{(m+i)}_{i, m+i+1}$.
			\item For $t \in [0, i-1]$, for all $P \in S^{(m+i)}_t$, we have
			\[
			v_P(f[i]) = \begin{cases} -r &\mbox{if $t \le \frac{m+i}{2}$} \\ -rq^{2t - (m+i)} &\mbox{if $t > \frac{m+i}{2}$.} \end{cases}
			\]
			\item $\deg \left((f[i])_\infty^{(m+i)}\right) = rq^i$.
		\end{itemize}
	\end{enumerate}
\end{myprop}
\begin{proof}
	For all $j$, we have the isomorphism
	\begin{align*}
		\phi_j: F_j &\xra{\sim} F_j \\
		x_k &\mapsto x_{j-k}^{-1},
	\end{align*}
	which is its own inverse \cite[p.\ 2237]{ADKSS}.  It is easy to see that
	\[
	f[i] = \phi_{m+i}(\phi_m(f)),
	\]
	proving (a).
	
	To prove (b), we use the fact that $\phi_j$ induces bijections $S^{(j)}_t \leftrightarrow S^{(j)}_{j-t}$ for each $t \in [0, j]$, together with a correspondence $P^{(j)}_\infty \leftrightarrow S^{(j)}_{j+1}$ \cite[p.\ 2237]{ADKSS}.  Thus letting $f \in F_m$ be regular of weight $r$, $\phi_m(f)$ is regular at all places (including $P^{(j)}_\infty$) except for a pole of order $r$ at $S^{(j)}_{j+1}$.  Then by the ramification behavior of the tower, $\phi_m(f) \in F_{m+i}$ is regular at all places (including $P^{(j)}_\infty$) except for $S^{(m+i)}_{m, m+i+1}$.  In particular, for $t \in [m, m+i+1]$, for all $P \in S^{(m+i)}_t$, we have
	\[
	v_P(\phi_m(f)) = \begin{cases} -r &\mbox{if $t \ge \frac{m+i}{2}$} \\ -rq^{m+i-2t} &\mbox{if $t < \frac{m+i}{2}$} \end{cases}
	\]
	Then $f[i] = \phi_{m+i}(\phi_m(f))$ easily has the first three properties in (b).  The fourth property follows either from computing the total pole degree directly, or from using the above properties of the conorm map to compute
	\begin{align*}
		\deg \left((f[i])^{(m+i)}_\infty \right)
		&= \deg \left((\phi_m(f))^{(m+i)}_\infty\right) \\
		&= [F_{m+i}: F_m] \deg \left((\phi_m(f))^{(m)}_\infty\right) \\
		&= q^i \deg \left((f)^{(m)}_\infty\right) \\
		&= r q^i.
	\end{align*}
\end{proof}

\subsection{Construction of Mostly Regular Functions}

\noindent We begin by dealing with the case $r \in [0, q^n - 1]$.  Write $r = r_1q^{n - \lceil n/k \rceil} + r_2q^{n - 2\lceil n/k \rceil} + \dots + r_{k-1}q^{n - (k-1)\lceil n/k \rceil} + r_{k}$, with $r_1, \dots, r_{k-1} \in [0, q^{\lceil n/k \rceil} - 1]$ and $r_{k} \in [0, q^{n - (k-1)\lceil n/k \rceil} - 1]$.  For $i \in [1, k-1]$, let $\bar{f}_i \in F_{\lceil n/k \rceil}$ be regular of weight $q^{\lceil n/k \rceil + 1} + r_i$, and let $\bar{f}_k \in F_{n - (k-1)\lceil n/k \rceil}$ be regular of weight $q^{n - (k-1)\lceil n/k \rceil + 1} + r_k$.  Set $$f_r := \prod_{i=1}^{k} \bar{f}_i[(i-1)\lceil n/k \rceil]. $$
The following proposition shows that $f_r$ is mostly regular with weight precisely $r + \sum_{i=1}^{k} q^{n - (i-1)\lceil n/k \rceil + 1}$.
\begin{myprop}\label{G_bound}
	There exists a pole divisor $G$ of degree at most $q^n(kq + k - 1)$ and weight $\sum_{i=1}^{k} q^{n - (i-1)\lceil n/k \rceil + 1}$ such that for all $r \in [0, q^n - 1]$, $$f_r  \in \mc{L}(G + r(P^{(n)}_\infty)) \setminus \mc{L}(G + (r-1)(P^{(n)}_\infty)).$$
\end{myprop}
\begin{proof}
Using proposition \ref{prop_shift}, it is easy to see that $f_r$ has weight $r + \sum_{i=1}^{k} q^{n - (i-1)\lceil n/k \rceil + 1}$ and is regular outside of $S^{(n)}_{0, (k-1)\lceil n/k \rceil - 1}$.  It remains to bound the pole orders at the places in $S^{(n)}_{0, (k-1)\lceil n/k \rceil - 1}$.

Let $\bar{f}_1, \dots, \bar{f}_k$ be as in the definition of $f_r$.  For $i \in [1, k-1]$, $\bar{f}_i$ has weight less than $q^{\lceil n/k \rceil + 1} + q^{\lceil n/k \rceil}$ in $F_{\lceil n/k \rceil}$.  Thus by the above proposition, the pole divisor of $\bar{f}_i[(i-1)\lceil n/k \rceil]$ in $F_{i\lceil n/k \rceil}$ satisfies
\[
(\bar{f}_i[(i-1)\lceil n/k \rceil])^{(i\lceil n/k \rceil)}_\infty - (q^{\lceil n/k \rceil + 1} + r_i)P^{(i\lceil n/k \rceil)}_\infty \le (q+1)q^{\lceil n/k \rceil}\sum_{t=0}^{(i-1)\lceil n/k \rceil - 1} q^{\max\{0, 2t - i\lceil n/k \rceil\}} S^{(i\lceil n/k \rceil)}_t.
\]
Here we use $S^{(i\lceil n/k \rceil)}_t$ as a shorthand for the divisor $\sum_{P \in S^{(i\lceil n/k \rceil)}_t} P$.  Let $G_i$ denote the divisor on the right-hand side.  By direct computation or by the same trick used in the proof of Proposition \ref{prop_shift}, we have $\deg(G_i) = (q+1)(q^{i\lceil n/k \rceil} - q^{\lceil n/k \rceil})$.

Next, $\bar{f}_k$ has weight less than $q^{n - (k-1)\lceil n/k \rceil + 1} + q^{n - (k-1)\lceil n/k \rceil}$ in $F_{n - (k-1)\lceil n/k \rceil}$.  Hence again
\[
(\bar{f}_k[(k-1)\lceil n/k \rceil])^{(n)}_\infty - (q^{n - (k-1)\lceil n/k \rceil + 1} + r_k)P^{(n)}_\infty \le (q+1)q^{n - (k-1)\lceil n/k \rceil}\sum_{t=0}^{(k-1)\lceil n/k \rceil - 1} q^{\max\{0, 2t - n\}} S^{(i\lceil n/k \rceil)}_t.
\]
Let $G_k$ denote the divisor on the right-hand side.  As above, we have $\deg(G_k) = (q+1)(q^n - q^{n - (k-1)\lceil n/k \rceil})$.

Using the pole divisors computed above, it is easy to see that
\[
(f_r)^{(n)}_\infty - v_{P^{(n)}_\infty}(f_r)P^{(n)}_\infty \le \sum_{i = 1}^{k-1} \Con^{F_n}_{F_{i\lceil n/k \rceil}} G_i + G_k.
\]
Define
\[
G = \left(\sum_{i=1}^{k} q^{n - (i-1)\lceil n/k \rceil + 1}\right) P^{(n)}_\infty + \sum_{i = 1}^{k-1} \Con^{F_n}_{F_{i\lceil n/k \rceil}} G_i + G_k.
\]
Then the above remarks show that
\[
f_r \in \mc{L}(G + r(P^{(n)}_\infty)) \setminus \mc{L}(G + (r-1)(P^{(n)}_\infty)),
\]
and
\begin{align*}
	\deg(G)
	&= \sum_{i=1}^{k} q^{n - (i-1)\lceil n/k \rceil + 1} + \sum_{i=1}^{k-1} (q+1)(q^n - q^{n - (i-1)\lceil n/k \rceil}) + (q+1)(q^n - q^{n - (k-1)\lceil n/k \rceil}) \\
	&\le k(q+1)q^n - q^n = q^n(kq + k - 1).
\end{align*}
\end{proof}
For general $r \ge 0$, say $r = s q^n + t$ with $t \in [0, q^n - 1]$, set $$f_r := x_0^s f_t.$$  Then we again have $f_r \in \mc{L}(G + r(P^{(n)}_\infty)) \setminus \mc{L}(G + (r-1)(P^{(n)}_\infty))$ because $x_0$ is regular of weight $q^n$.

\section{Code Sequences Beating the Gilbert-Varshamov Bound}\label{section_codeconstruction}

Define an $\F_{q^2}$-linear map $\psi: \F_{q^2}^{\N^0} \ra \bigcup_r \mc{L}(G + rP^{(n)}_\infty)$ by sending the $r$-th basis vector to $f_r$ (we zero-index the basis vectors).  Because the $f_r$ have distinct weights, the strict triangle inequality implies that $\psi$ is injective.

Each $f_r$ is regular at all of the code places.  Hence we can speak of the evaluation map $\ev: \bigcup_r \mc{L}(G + rP^{(n)}_\infty) \ra \F_{q^2}^{q^n(q^2 - q)}$, which maps a function to the tuple of its values at the code places.

\begin{myprop}\label{prop_existence}
Let $K \in [1, q^n(q^2 - q - kq - k + 1)]$.  Then the map $(\ev \circ \psi)|_{\F_{q^2}^K}: \F_{q^2}^K \ra \F_{q^2}^{q^n(q^2 - q)}$ is injective, and its image defines an $[N, K, D]$ code, where $N = q^n(q^2 - q)$ and
\[
D \ge N - K - q^n(kq + k - 1) + 1.
\]
\end{myprop}
\begin{proof}
This follows from the above bound on $\deg(G)$ and a standard argument about algebraic geometry codes.  For completeness, we give the proof in full.

Let $D^* := N - K - q^n(kq + k - 1) + 1$.  Suppose that for some nonzero $v \in \F_{q^2}^K$, the $N$-tuple $\ev(\psi(v))$ has less than $D^*$ nonzero coordinates.  Let $M > N - D^*$ be the number of coordinates which are zero.  We already know that $\psi(v) \in \mc{L}(G + (K-1)P^{(n)}_\infty)$.  By definition of $M$, there are code places $P_1, \dots, P_M$ at which $\psi(v)$ is zero.  Then $\psi(v)$ lies in the Riemann-Roch space $\mc{L}(D)$, where
\[
D = G + (K-1)P^{(n)}_\infty - \sum_{i=1}^M P_i.
\]
But $\deg(D) = \deg(G) + K - 1 - M < q^n(kq + k - 1) + K - 1 - N + D^* = 0$ by Proposition \ref{G_bound} and the definition of $D^*$, so $\mc{L}(D) = \{0\}$ and $\psi(v) = 0$.  But we said above that $\psi$ is injective, so this is a contradiction.

To see that $(\ev \circ \psi)|_{\F_{q^2}^K}$ is injective, note that for $K \le q^n(q^2 - q - kq - k + 1)$, we have $D^* \ge 1$, hence the above argument shows that at least one coordinate of $\ev(\psi(v))$ is nonzero.
\end{proof}

The above proposition implies that for any $n \ge 0$, we can define codes of the above form with length $q^n(q^2 - q)$ over $\F_{q^2}$ whose rate $R$ and relative distance $\delta$ satisfy
\[
R + \delta \ge 1 - \frac{kq + k - 1}{q^2 - q},
\]
with many choices of rate.  For all $k$, this exceeds the Gilbert-Varshamov bound for large enough $q$.

\section{Subquadratic Time Encoding}\label{section_encoding}
The encoding task is: given a message $v \in \F_{q^2}^{q^n(q^2 - q - kq - k + 1)}$, output $\ev(\psi(v))$. For simplicity, we assume throughout this section that $k$ divides $n$; this affects the runtime by a factor of at most $\mbox{poly}(q)$, which is a constant in our context.  Our goal in this section is to prove the following result.

\begin{mythm}\label{thm_main_encoding}
Assume $k \mid n$.  For the codes described in Proposition \ref{prop_existence}, there exist deterministic algorithms to:
\begin{itemize}
  \item pre-compute a representation of the code at the encoder using $O((n/k)^3q^4(q^n)^3)$ operations over $\F_q$; this representation occupies $O(q^2(q^n)^{2/k}\log q)$ space
  
  \item encode a message using $O(kq^4(q^n)^{1+\frac{\omega - 2}{k}})$ operations over $\F_q$, where $\omega$ is the matrix multiplication exponent.
\end{itemize}
\end{mythm}
Taking $k = 2$ implies the encoding portion of Theorem \ref{theorem_main_2}, noting that there we treat $q$ as a constant and instead take $N = q^n(q^2 - q)$ to be the parameter of interest.  See Table \ref{k_tradeoff} for a comparison of these runtimes for $k = 2, 3$.

Our approach is to first write the encoding of a vector $w \in \F_{q^2}^{q^{i(n/k)}}$ with respect to $F_{i(n/k)}$ in terms of some encodings with respect to $F_{(i-1)(n/k)}$ and $F_{n/k}$.  We then use a Baby-Step Giant-Step algorithm and fast matrix multiplication to build up the encoding of $v$ starting from encodings with respect to $F_{n/k}$.



\subsection{Pre-computation}\label{pre-computation}
Pre-compute the evaluations of some $g_0, \dots, g_{q^{n/k} - 1} \in F_{n/k}$ at the code places of $F_{n/k}$, where each $g_s$ is regular of weight $q^{n/k + 1} + s$.  This can be done using the deterministic algorithm in \cite{ADKSS}.  We then need to store $O(q^{n/k}\cdot q^{n/k}(q^2 - q)) = O(q^2(q^n)^{2/k})$ elements of $\F_q$.

\subsection{Subquadratic Time Encoding with Fast Matrix Multiplication}

We begin by considering encoding for $v \in \F_{q^2}^{q^n}$. Encoding messages of length greater than $q^n$ will be dealt with at the end of this subsection.

For $i \in [0, k]$, $w \in \F_{q^2}^{q^{i(n/k)}}$, and $P$ a code place of $F_{i(n/k)}$, we set $w(P) := \psi_i(w)(P)$, where $\psi_i$ is the function $\psi$ corresponding to $F_{i(n/k)}$.

\begin{myprop}\label{encoding_sum}
Let $i \in [1, k]$, let $w \in \F_{q^2}^{q^{i(n/k)}}$, and let $P$ be a code place of $F_{i(n/k)}$.  Uniquely write
\[
w = \sum_{\ell=0}^{q^{n/k} - 1} \iota_\ell(w^{(\ell)})
\]
for $w^{(\ell)} \in \F_{q^2}^{q^{(i-1)(n/k)}}$, where $\iota_\ell: \F_{q^2}^{q^{(i-1)(n/k)}} \hookrightarrow \F_{q^2}^{q^{i(n/k)}}$ is the vector space embedding sending the $j$-th basis vector to the $(j + \ell q^{(i-1)(n/k)})$-th basis vector.  Let $P'$ denote the place obtained by restricting $P$ to $F_{(i-1)(n/k)}$, and let $P''$ denote the place of $F_{n/k}$ at which $x_0$ has value $x_{(i-1)(n/k)}(P)$, $x_1$ has value $x_{(i-1)(n/k) + 1}(P)$, etc.  Then $P'$ and $P''$ are code places, and
\[
w(P) = \sum_{\ell=0}^{q^{n/k} - 1} w^{(\ell)}(P')g_\ell(P'').
\]
\end{myprop}
\begin{proof}
By \cite[Lemma 3.9]{GS}, the code places of $F_m$ are precisely the places at which each $x_t$ has value in $\F_{q^2} \setminus \Omega$, subject to the relations defining the Garcia-Stichtenoth tower.  Both $P'$ and $P''$ have this form, so they are code places.

Next, by the definition of the $f_r$, it is easy to see that
\[
\psi_i(w) = \sum_{\ell=0}^{q^{n/k} - 1} \psi_{i-1}(w^{(\ell)}) \left(g_\ell[(i-1)(n/k)]\right).
\]
The claimed equation follows immediately.
\end{proof}

Observe that $w(P)$ looks like an element of a product matrix.  We can write down an explicit matrix product as follows.

Fix $i \in [1, k]$ and $\alpha \in \F_{q^2} \setminus \Omega$.  Let $\mc{P}_\alpha$ denote the set of code places $P$ of $F_{i(n/k)}$ for which $x_{(i-1)(n/k)}(P) = \alpha$, let $\mc{P}'_\alpha$ denote the set of code places $Q$ of $F_{(i-1)(n/k)}$ such that $x_{(i-1)(n/k)}(Q) = \alpha$, and let $\mc{P}''_\alpha$ denote the set of code place $R$ of $F_{n/k}$ such that $x_0(R) = \alpha$.  Then it is easy to see that for any $Q \in \mc{P}'_\alpha$ and $R \in \mc{P}''_\alpha$, there is a unique place $P \in \mc{P}_\alpha$ such that $P' = Q$ and $P'' = R$, where $P'$ and $P''$ are as in the above proposition.  Conversely, if $P \in \mc{P}_\alpha$, then $P' \in \mc{P}'_\alpha$ and $P'' \in \mc{P}''_\alpha$.

Easily $|\mc{P}'_\alpha| = q^{i(n/k)}$ and $|\mc{P}''_\alpha| = q^{n/k}$.  Let $Q_1^\alpha, \cdots, Q_{q^{i(n/k)}}^\alpha$ be an enumeration of $\mc{P}'_\alpha$, and let $R_1^\alpha, \cdots, R^\alpha_{q^{n/k}}$ be an enumeration of $\mc{P}''_\alpha$.

\begin{myprop}\label{mat_mult_lemma}
Let $i \in [1, k]$, and let $w \in \F_{q^2}^{q^{i(n/k)}}$.  Write $
w = \sum_{\ell=0}^{q^{n/k} - 1} \iota_\ell(w^{(\ell)})$
as in Proposition \ref{encoding_sum}.  For each $\alpha \in \F_{q^2} \setminus \Omega$, define a matrix $A^\alpha$ of shape $q^{(i-1)(n/k)} \times q^{n/k}$ and a matrix $B^\alpha$ of shape $q^{n/k} \times q^{n/k}$ by
\begin{align*}
	&A^\alpha_{st} = w^{(t - 1)}(Q^\alpha_s)
	&B^\alpha_{st} = g_{s-1}(R^\alpha_t).
\end{align*}
Then for every code place $P$ of $F_{i(n/k)}$, letting $\alpha = x_{(i-1)(n/k)}(P)$ and letting $s, t$ be such that $P' = Q^\alpha_s$ and $P'' = R^\alpha_t$, we have $w(P) = (A^\alpha B^\alpha)_{st}$.
\end{myprop}
\begin{proof}
This is just a restatement of Proposition \ref{encoding_sum}.
\end{proof}

Using this proposition, it is not too difficult to define an algorithm \textsc{Matrix-Encode} which encodes $v \in \F_{q^2}^{q^n}$ using a series of $k(q^2 - q)$ matrix multiplications of shape
\begin{equation}\label{mult_shape}
	\left(q^{n - n/k} \times q^{n/k}\right) \times \left(q^{n/k} \times q^{n/k}\right).
\end{equation}

\renewcommand{\figurename}{Algorithm}
\begin{figure}[H]
	\begin{algorithmic}[5]
		%
		\Procedure{Matrix-Encode}{$v \in \F_{q^2}^{q^n}$}
		\State $W_k \gets \{v\}$
		\For{$i$ from $k$ to 1}
		\State $W_{i-1} \gets \emptyset$
		\For{$w \in W_i$}
		\State Write $w = \sum_{\ell=0}^{q^{n/k} - 1} \iota_\ell(w^{(\ell)})$ as in Proposition \ref{encoding_sum}
		\State Add all $w^{(\ell)}$ to $W_{i-1}$
		\EndFor
		\EndFor
		\For{$i$ from 1 to $k$}
		\For{$\alpha \in \F_{q^2} \setminus \Omega$}
		\For{$w \in W_i$}
		\State Construct the matrix $A^\alpha_w$ corresponding to $w$ in Proposition \ref{mat_mult_lemma}, using the $w^{(\ell)}(Q)$ computed in iteration $i - 1$ (when $i = 1$, just use the values of the scalars $w^{(\ell)}$)
		\EndFor
		\State Let $\bar{A}^\alpha$ be the matrix made of all $A^\alpha_w$ stacked vertically, and let $B^\alpha$ be as in Proposition \ref{mat_mult_lemma}
		\State Multiply $\bar{A}^\alpha$ by $B^\alpha$, thus computing $w(P)$ for all $w \in W_i$ and code places $P$ of $F_{i(n/k)}$
		\EndFor
		\EndFor
		
		\EndProcedure
	\end{algorithmic}
	\caption{The algorithm \textsc{Matrix-Encode}.  It inputs $v \in \F_{q^2}^{q^n}$ and outputs $\ev(\psi(v))$.}
	\label{matrix_encode}
\end{figure}

\textbf{Encoding for messages of length longer than $q^n$:} To encode $v \in \F_{q^2}^{q^n(q^2 - q - kq - k + 1)}$, we just need to make $q^2 - (k+1)q$ calls to \textsc{Matrix-Encode}, using the fact that $f_{s q^n + t} = x_0^s f_t$.

\subsection{Complexity of the Encoding Algorithm}\label{appendix_complexity}
Performing the pre-computation using the algorithm in \cite{ADKSS} requires at most $(n/k)^3q^{3n/k + 4} = O((n/k)^3q^4 (q^n)^{3n/k})$ multiplications and divisions in $\F_{q^2}$.  Storing the evaluations requires space $O(q^{2n/k + 2}) = O(q^2(q^n)^{2/k})$, since there are $q^{n/k}$ functions $g_s$ and $O(q^{n/k + 2})$ code places of $F_{n/k}$.


Next, we compute the runtime of the encoding function for general $v \in \F_{q^2}^{q^n(q^2 - q - kq - k + 1)}$.  The runtime of each call to \textsc{Matrix-Encode} is just the runtime of $k(q^2 - q)$ matrix multiplications of shape (\ref{mult_shape}).  Then the runtime to encode general $v \in \F_{q^2}^{q^n(q^2 - q - kq - k + 1)}$ is the runtime of $k(q^2 - q)^2 = O(kq^4)$ such multiplications.  Thus using fast square matrix multiplication, we can encode $v$ using
\[
O\left(kq^4q^{n - 2n/k}(q^{n/k})^\omega \right) = O\left(kq^4(q^n)^{1 + \frac{\omega - 2}{k}}\right) 
\]
operations over $\F_{q^2}$, where $\omega$ is the exponent of (square) matrix multiplication.  This complete the proof of Theorem \ref{thm_main_encoding}.

Using the best known bound $\omega \le 2.37$ \cite{lG}, we attain the runtime $O\left(kq^4(q^n)^{1 + 0.37/k}\right)$.  When $k \ge 3$, we can instead use fast $(M^2 \times M) \times (M \times M)$ rectangular matrix multiplication; letting $\omega'$ be the exponent of such multiplication, we get an algorithm running in time $O(kq^4(q^n)^{1 + (\omega' - 3)/k})$. We compare the parameters for $k = 2$ and 3 in Table \ref{k_tradeoff} below. 

\begin{table}[h]
\centering
\begin{TAB}(@){c|c|c}{c|c|c|c|c|c}
$\bm{k}$ & \textbf{2} & \textbf{3} \\
\textbf{Preprocessing time} & $O(N^{1.5}\log^3 N)$ & $O(N\log ^3N)$ \\
\textbf{Encoding time with $\bm{\omega = 2}$} & $O(N)$ & $O(N)$ \\ 
\textbf{Encoding time with $\bm{\omega \approx 2.37}$} & $O(N^{1.19})$ & $O(N^{1.13})$ \\
\textbf{Encoding time with $\bm{\omega' \approx 3.34}$} & N/A & $O(N^{1.12})$ \\ 
\textbf{Smallest $\bm{q}$ beating Gilbert-Varshamov bound} & 19 & 32 \\ 
\end{TAB}
\caption{Comparison of encoding times and code quality for $k = 2, 3$.  Here $N$ is the code length, $\omega$ is the exponent of square matrix multiplication, and $\omega'$ is the exponent of $(M^2 \times M) \times (M \times M)$ rectangular matrix multiplication. Runtime dependence on $q$ is absorbed in the asymptotic notation as $q$ is a constant for each family of codes.}
\label{k_tradeoff}
\end{table}

\subsection{Systematic Subcodes for Interactive Oracle Proofs}\label{subsection_iop}
We sketch a construction of systematic subcodes that lowers the efficiency exponent of the IOPs in \cite{bcgrs}[Lem 7.2] from $c>3$ to $c>3/2$. For a positive integer $n$, let $G_n$ be the generator matrix for an instance of our code over $F_n$ which has rate and relative distance at least $1/4$; this exists so long as $q$ is sufficiently large as a function of $k$.  For any particular $\alpha \in \F_{q^2} \setminus \Omega$, there are $q^n$ rows of $G_n$ corresponding to points with $x_{n/2} = \alpha$; let $H$ denote $G_n$ restricted to these rows.  Then after rearranging rows and columns, we have $H = A \otimes B$, where $A$ is the restriction of $G_{n/2}$ to rows corresponding to points with $x_{n/2} = \alpha$, and $B$ is the restriction of $G_{n/2}$ to rows corresponding to points with $x_0 = \alpha$.  Letting $A$ and $B$ have ranks $r_1$ and $r_2$, by column reducing $G_{n/2}$ (in different ways for $A$ and $B$), we can take $A$ and $B$ to begin with the diagonal blocks $I_{r_1}$ and $I_{r_2}$.  Thus after applying column operations to $G_n$, we can assume that $H$ begins with the diagonal block $I_{r_1 r_2}$.  It follows that there is a systematic subcode of $G_n$'s code with dimension $r_1r_2$ and relative distance at least $1/4$.  Furthermore, this subcode can be encoded in time $O(N^{\omega/2})$ with preprocessing time $O(N^{3/2}\log^3 N)$, as with our original code.

It remains to show that we can always choose $\alpha$ so that $r_1r_2$ is a positive constant fraction of $N$.  The sum of $r_1r_2$ across all $\alpha$ is just the rank of $G_n$, which is $K \ge N/4$.  Thus there exists $\alpha$ for which $r_1r_2/N \ge 1/(4q^2)$, as desired.  We can find such an $\alpha$ during the preprocessing step in time $O(N^{3/2})$, since computing $r_1r_2$ for each $\alpha$ just requires finding the ranks of $2(q^2 - q)$ matrices of size $q^{n/2} \times q^{n/2}$.



\section{Fast Decoding Algorithms}\label{section_decoding}
Reed-Solomon codes are widespread in practice partly due to 
fast algebraic decoding algorithms: the Gorenstein-Zierler decoder, rational approximation using the Euclidean algorithm, the Berlekamp-Massey algorithm and fast Fourier decoders, to name a few. In particular, Reed-Solomon codes can be uniquely decoded in linear time up to the unique decoding limit. In a breakthrough, Sudan designed an algorithm to list decode Reed-Solomon codes beyond half the minimum distance \cite{sud}. List decoding is a relaxation of unique decoding where the decoder is allowed to output a list of messages and is deemed successful if the message sent is in the list. Shokrollahi and Wasserman soon generalized Sudan's algorithm to algebraic geometry codes \cite{SW}. Shortly thereafter, Guruswami and Sudan designed list decoders for both Reed-Solomon codes and algebraic geometry codes that improved on the error correction of previously known algorithms \cite{GS2}. A novelty they introduced was to use multiplicities in the interpolation step. 

We present a unique decoding algorithm that corrects a fraction of errors close to half the relative distance and a list decoding algorithm to correct beyond that. The algorithms are presented as specializations of the Shokrollahi-Wasserman algorithm to the Garcia-Stichtenoth tower. In particular, we obtain the unique decoding algorithm as a special case of the list decoding algorithm. 

Our results are as follows.

\begin{mythm}\label{thm_main_decoding}
For the codes described in Proposition \ref{prop_existence} with dimension $K$ and length $N = q^n(q^2 - q)$, there exist randomized Las Vegas algorithms to:
\begin{itemize}
  \item pre-compute a representation of the code at the decoder in $O(N^{3/2}\log^3 N)$ time; this representation occupies $O(N^{2/k})$ space
  
  \item uniquely decode up to $\frac{1}{2}\left(N - K - 1 - 4\frac{kq + k - 1}{q^2 - q}N\right)$ errors in $O(N^{2 + (\omega - 2)/k} \log^2 N)$ expected time, where $\omega$ is the matrix multiplication exponent
  
  \item for the list decoding algorithm, additionally pre-compute a matrix for the lifting step in $O(N^\omega)$ time, occupying $O(N^2)$ space
  
  \item list decode up to $N \left(1-\sqrt{2\left(R+2\frac{kq+k-1}{q(q-1)}\right)}-2\frac{kq+k-1}{q(q-1)} \right)$ errors with list size at most $\sqrt{\frac{2(q-1)}{2 + R(q-1)}}$ in $O(N^{2 + (\omega - 2)/k} \log^2 N)$ expected time.
\end{itemize}
\end{mythm}
Setting $k = 2$ yields the decoding portion of Theorem \ref{theorem_main_2}.

\subsection{Pre-computing a Representation of the Code} This pre-computation step is the same as for the encoding algorithm (see Section \ref{pre-computation}) and is deterministic.  It is needed so that we may call our encoding algorithm as a sub-routine.  Additional pre-computation for the list decoding algorithm is discussed in Section \ref{root_finding}.

\subsection{Modified Shokrollahi-Wasserman Algorithm}
Consider codes of block length $N=q^n(q^2-q)$ and dimension $K$ constructed in \S4 with parameter $k$. Let $P_1, P_2, \dots, P_N$ denote the code places, which are places of $F_n$, and let $y=(y_1,y_2,\ldots,y_N) \in \F_{q^2}^N$ denote the received word, where $y_i$ is the (possibly errored) evaluation at $P_i$.  Let $\ell$ be a bound on the number of messages allowed in the list. Let $B$ be an agreement parameter (determined later), that is, we need to correct fewer than $N-B$ errors. The algorithm first interpolates a nonzero polynomial $$H(T):=u_0+u_1T+\ldots+u_{\ell -1}T^{\ell -1}+u_\ell T^\ell  \in F_n[T]$$ in an indeterminate $T$ such that every message in the list (that is, every message whose encoding agrees with the received word at more than $B$ evaluation places) is a root of $H(T)$. Then the roots of $H(T)$ that are in the message Riemann-Roch space $\mc{L}(G+(K-1)P^{(n)}_\infty)$ are enumerated as a list of candidate messages. The encoding algorithm is finally used to check which messages sufficiently agree with the received word and indeed belong in the list.

To construct such a polynomial $H(T)$, we insist
\begin{equation}\label{equation_interpolation1}
H(y_i)(P_i) =\left( \sum_{j=0}^{\ell} u_j(P_i)y_i^j\right) = 0, \forall i\in\{1,2,\ldots,N\}
\end{equation} 
and that 
\begin{equation}\label{equation_interpolation2}
u_j \in \mc{L}(G+w_jP^{(n)}_\infty), \forall j \in \{0,1,\ldots,\ell\} 	
\end{equation}
where $w_j:= B-(\ell+1)\deg(G)-(K-1)j$. The latter constraint $u_j \in \mc{L}(G+w_jP^{(n)}_\infty)$ ensures that when we substitute a function $f \in \mc{L}(G+(K-1)P^{(n)}_\infty)$, the resulting function $H(f)$ lies in $$\mc{L}\left((\ell+1)G + (B - (\ell+1)\deg(G))P^{(n)}_\infty\right),$$ hence has a pole divisor of degree at most $B$.

\noindent Suppose $\ell $ and $B$ are such that $B \ge (N+1)/(\ell +1) + \ell (K - 1 + 2\deg(G))/2 + \deg(G) - 1$.  We claim that we can construct a nonzero polynomial $H(T)$ satisfying the constraints.  We attempt to populate each coefficient space $\mc{L}(G+w_jP^{(n)}_\infty)$ with enough of the functions constructed in \S3. In particular, we enforce the second constraint (equation \ref{equation_interpolation2}) by insisting $u_j$ be in the span of $\{f_0,f_1,\ldots,f_{w_j}\} \subseteq \mc{L}(G+w_jP^{(n)}_\infty)$. (If $w_j < 0$, we take $u_j = 0$.) Writing each $u_j$ as an unknown linear combination of $\{f_0,f_1,\ldots,f_{w_j}\}$, the first constraint (equation \ref{equation_interpolation1}) is an $\F_{q^2}$-linear system in at least $w_0+w_1+\ldots+w_\ell+\ell+1$ variables and $N$ constraints.  The condition on $\ell , B$ ensures
\begin{align*}
w_0+w_1+\ldots+w_\ell+\ell+1 = (\ell +1)(B - (\ell +1)\deg(G) - \ell (K-1)/2 + 1) \ge N+1 > N,
\end{align*}
proving our claim.

Henceforth, fix $\ell := \lfloor \sqrt{2N/(K - 1 + 2\deg(G))}\rfloor$ and $B:=\lceil \sqrt{2N(K - 1 + 2\deg(G))} + \deg(G) \rceil$.  To prove that we can construct $H(T)$ satisfying the constraints for these values of $\ell$ and $B$, observe that
\begin{align*}
&(N+1)/(\ell +1) + \ell (K - 1 + 2\deg(G))/2 + \deg(G) - 1 \\
&\le N/(\ell +1) + \ell (K - 1 + 2\deg(G))/2 + \deg(G) \\
&\le \frac{N}{\sqrt{2N/(K - 1 + 2\deg(G))}} + \frac{\sqrt{2N(K - 1 + 2\deg(G))}}{2} + \deg(G) \\
&\le B,
\end{align*}
hence the above claim applies.


Say $f \in \mc{L}(G+(K-1)P^{(n)}_\infty)$ agrees with $(y_1,y_2,\ldots,y_N)$ at more than $B$ places. Then $H(f)$ has a zero at more than $B$ places yet pole degree at most $B$. Hence $H(f)=0$ and $f$ is indeed a root of $H(T)$. Thus we can tolerate fewer than $N-B$ errors.    

From the proof of proposition \ref{G_bound}, $1/(q-1) < \deg(G)/N \le \frac{kq+k-1}{q(q-1)}$. 
Thus with list size at most $\sqrt{\frac{2(q-1)}{2 + R(q-1)}}$,
we decode up to $N \left(1-\sqrt{2\left(R+2\frac{kq+k-1}{q(q-1)}\right)} - \frac{kq+k-1}{q(q-1)}\right)$ errors.

In the subsequent subsections, we show how to interpolate $H(T)$ and find its roots in the message space in subcubic time. Our code constructions are stated (for instance in Theorem \ref{theorem_main_2}) as families of codes of increasing block length $N$  for each $q$ and $R$. Hence $q$ and $R$ shall be treated as constants independent of $N$ in the subsequent complexity estimates. In particular, the list size bound $\ell$ will be treated as a constant independent of $N$.

\subsection{Fast Interpolation using Black-Box Linear Algebra}
Consider the matrix 
\begin{equation*}
	M:=\begin{bmatrix}
	f_0(P_1) & \ldots & f_{w_0}(P_1) & y_1 f_0(P_1) & \ldots & y_1 f_{w_1}(P_1) & \ldots\ldots\ldots  & y_1^\ell f_0(P_1)  & \ldots & y_1^\ell f_{w_\ell}(P_1)\\
	f_0(P_2) & \ldots & f_{w_0}(P_2) & y_2 f_0(P_2) & \ldots & y_2 f_{w_1}(P_2) & \ldots\ldots\ldots  & y_2^\ell f_0(P_2)  & \ldots & y_2^\ell f_{w_\ell}(P_2)\\
	\vdots & \ddots & \vdots & \vdots & \ddots & \vdots & \ddots\ddots\ddots  & \vdots  & \ddots & \vdots\\
	f_0(P_N) & \ldots & f_{w_0}(P_N) & y_N f_0(P_N) & \ldots & y_N f_{w_1}(P_N) & \ldots\ldots\ldots  & y_N^\ell f_0(P_N)  & \ldots & y_N^\ell f_{w_\ell}(P_N)
	\end{bmatrix}
\end{equation*}
corresponding to the linear system in equation \ref{equation_interpolation1}. Given a column vector $$a:=(a_{0,0},\ldots,a_{0,w_0},a_{1,0},\ldots,a_{1,w_1},\ldots\ldots,a_{\ell,0},\ldots,a_{\ell,w_\ell})^t$$
over $\F_{q^2}$, the matrix-vector product 
$$Ma= \left(\sum_{j=0}^{w_0}a_{0,j}f_j(P_i)+y_i\sum_{j=0}^{w_1}a_{1,j}f_j(P_i)+\ldots+y_i^\ell\sum_{j=0}^{w_\ell}a_{\ell,j}f_j(P_i),\ i=1,2,\ldots,N\right)^t$$
can be computed in subquadratic time as follows. For each $b \in \{0,1,\ldots,\ell\}$, 
$$\sum_{j=0}^{w_b}a_{b,j}f_j(P_i), \forall i\in \{1,2,\ldots,N\}$$ can be computed in $O(N^{1 + (\omega-2)/k})$ time using the encoding algorithm in \S~\ref{section_encoding} by viewing it as encoding the function $\sum_{j=0}^{w_b}a_{b,j}f_j$. Taking the inner product with  $(y_1^b,y_2^b,\ldots,y_N^b)^t$ yields $$y_i^b\sum_{j=0}^{w_b}a_{b,j}f_j(P_i), \forall i\in \{1,2,\ldots,N\}$$ in linear time. Adding up these terms obtained for $b\leq \ell$ yields $Ma$ in $O(N^{1 + (\omega-2)/k})$ time. We invoke Wiedemann's algorithm \cite{Wie} to find a nonzero solution to $Mx=0$ and obtain the interpolation polynomial. Since matrix-vector products take $O(N^{1 + (\omega-2)/k})$ time, the linear system is solved in $O(N^{2+(\omega-2)/k}\log^2N)$ expected time.

\subsection{Unique Decoding}
We obtain our unique decoding algorithm by choosing $\ell=1$ instead of the above choice of $l$. Since $H(T)$ is now degree one, the root finding step is trivial and involves just one division of functions. For $\ell=1$, we may choose $B = \frac{1}{2}(N+K) + 2\deg(G)$, allowing us to correct $\frac{1}{2}(N-K-1) - 2\deg(G)$ errors. The unique decoding bound assures that the code can correct $(N-K-\deg(G)-1)/2$ errors, that is, a $(1-R-\deg(G)/N)/2$ fraction of errors. Our algorithm is guaranteed to correct a $(1-R-4\deg(G)/N)/2$ fraction of errors, falling short by the small term $(3/2)\deg(G)/N \leq (3/2)(kq+k-1)/(q^2-q)$. There are ways (analogous to \cite{fr,sjmjh}) to modify the algorithm and correct up to the unique decoding assurance. 
We refrain from detailing the changes since list decoding subsumes such improvements.

A question, which we leave open, is if unique decoding could be performed in subquadratic expected time. To this end, one might consider the matrix 
\begin{equation*}
M=\begin{bmatrix}
f_0(P_1) & f_1(P_1) & \ldots & f_{w_0}(P_1) & y_1 f_0(P_1) &y_1 f_1(P_1) & \ldots & y_1 f_{w_1}(P_1) \\
f_0(P_2) & f_1(P_2) & \ldots & f_{w_0}(P_2) & y_2 f_0(P_2) &y_2 f_1(P_2) & \ldots & y_2 f_{w_1}(P_2) \\
\vdots & \vdots & \ddots & \vdots & \vdots & \vdots & \ddots & \vdots \\
f_0(P_N) & f_1(P_N) & \ldots & f_{w_0}(P_N) & y_N f_0(P_N) &y_N f_1(P_N) & \ldots & y_N f_{w_1}(P_N) \\
\end{bmatrix}
\end{equation*}
corresponding to the linear system of the interpolation step. The bottleneck in unique decoding is computing a nonzero element in the null space of $M$. If $M$ were to have sublinear (that is, $o(N)$) displacement rank (see \cite{kkm} for definition), then by \cite{bjs} this task can be accomplished in sub quadratic time. Displacement ranks of interpolation matrices arising in list decoding were bounded by Olshevsky and Shokrollahi \cite[\S~5]{os}. However, their bounds apply only to codes from plane curves, and it is not immediate if their techniques imply sublinear displacement rank for $M$.

\subsection{Root Finding for List Decoding}\label{root_finding}
The main algorithmic challenge left in list decoding is root finding: to enumerate all the roots of $H(Y)$ in $\mc{L}(G+(K-1)P^{(n)}_\infty)$. Our strategy is to first pick a place $P$ in $F_n$ of degree just greater than $\deg(G)+K-1$, thus ensuring that the evaluation map from $\mc{L}(G+(K-1)P^{(n)}_\infty)$ to the residue field at $P$ is injective. Then find the roots of the reduction of $H(T)$ at $P$ and lift the roots to $\mc{L}(G+(K-1)P^{(n)}_\infty)$.

Such a $P$ can be found in nearly linear time by Artin-Schreier theory. Let $D = \deg(G)+K \le N$. Pick an $\alpha_0 \in \{\alpha \in \F_{q^{2D}}\ |\ \alpha^q+\alpha \neq 0\} \bigcap \left(\F_{q^{2D}}\setminus\F_{q^{2(D-1)}}\right)$.  To find such an $\alpha_0$, choose $\alpha_0$ to be a root of a random degree $D$ irreducible polynomial over $\F_{q^2}$ (in time nearly linear in $D$ using the algorithm of Couveignes and Lercier \cite{CL}). With probability at least $1-1/q \geq 1/2$, $\alpha_0^q+\alpha_0 \neq 0$. Once such an $\alpha_0$ is found, we look for a place in $F_n$ above $(x_0-\alpha_0)$. Observing $\alpha_0^q/(\alpha_0^{q-1}+1)$ is the fraction of the norm $\alpha_0^{q+1}$ and trace $\alpha_0^q+\alpha_0$ of $\alpha_0$ down to $\F_{q^{2(D-1)}}$, we see $\alpha_0^q/(\alpha_0^{q-1}+1) \in \F_{q^{2(D-1)}}$. Since the trace from $\F_{q^{2D}}$ down to $\F_{q^{2(D-1)}}$ is surjective, 
$$x_1^q+x_1 = \frac{\alpha_0^q}{\alpha_0^{q-1}+1}$$
has a solution $x_1 = \alpha_1 \in \F_{q^{2D}}$. Such an $\alpha_1$ can be found either using Hilbert's theorem 90 or using a generic root finding algorithm. The root finding algorithm of Kaltofen-Shoup \cite{KS}[Algorithm E, Theorem 1] implemented using the Kedlaya-Umans modular composition algorithm \cite{kedlaya_umans} takes expected time $O(N^{1 + o(1)})$. Further, $\alpha_1 \in \{\alpha \in \F_{q^{2D}}| \alpha^q+\alpha \neq 0\}$ since $\alpha_0^q/(\alpha_0^{q-1}+1) \neq 0$. We can thus iterate this process up the tower and find a place $P:=(\alpha_0:\alpha_1:\ldots:\alpha_n:1) $.  Since we insist $\alpha_0 \in \F_{q^{2D}} \setminus \F_{q^{2(D-1)}}$, the degree of $P$ is indeed $D$. The computation of the point $P$ may be moved to pre-processing (taking $O(N^{1 + o(1)})$ expected time and $O(N\log N)$ storage).

Given $P$, we reduce the coefficients of $H(T)$ modulo $P$. To perform this reduction in time quadratic in $N$, we first pre-compute and store $\{f_0(P),f_1(P),\ldots,f_{w_0}(P)\}$ in time $O(N^{1 + \omega/k}\log N + N^2\log N \log\log N)$ and storage space $O(N^2)$, as follows.  We assume the $f_r$'s are defined in terms of the basis for regular functions in $F_{n/k}$ returned by the algorithm in \cite{ADKSS}.  This algorithm writes each regular function in $F_{n/k}$ as an $\F_{q^2}$-linear combination of a fixed set of $O(N^{1/k}\log N)$ functions, each of which can be written as an $O(\log N)$ size arithmetic circuit in terms of $x_0, x_1, \dots, x_{n/k}$ (see \cite[Theorem 6]{ADKSS}).  We can evaluate all of these functions at each of the places $(\alpha_0 : \alpha_1 : \ldots : \alpha_{n/k} : 1), (\alpha_{n/k} : \alpha_{n/k+1} : \ldots : \alpha_{2n/k} : 1), \dots, (\alpha_{(k-1)n/k} : \alpha_{(k-1)n/k + 1} : \ldots : \alpha_n : 1)$ of $F_{n/k}$ in time $O(N^{1 + 1/k}\log^3 N\log\log N)$, noting that because all $\alpha_i$ are not roots of $\alpha^q + \alpha = 0$, using the obvious circuit will only require arithmetic operations over $\F_{q^{2D}}$ (i.e., no pole cancelling is required).  Representing these values as a $2D \times O(N^{1/k}\log N)$ matrix over $\F_{q^2}$, we can then multiply by the $O(N^{1/k}\log N) \times q^{n/k}$ matrix over $\F_{q^2}$ which writes the regular functions in $F_{n/k}$ in terms of these functions, yielding the evaluations of a basis for $F_{n/k}$ at each of the places above in time $O(N^{1+\omega/k}\log N)$.  To finish the pre-computation, we take the $w_0 + 1 = O(N)$ products corresponding to the definitions of the $f_r$'s at $P$, each of which involves $k$ multiplications over $\F_{q^{2D}}$; in total, this takes time $O(N^2\log N \log\log N)$.  

Once $H(T)$ is reduced, since its degree is a constant independent of $N$, all its roots can be enumerated in time nearly linear in $N$ using the Kaltofen-Shoup root finding algorithm \cite{KS} implemented using the Kedlaya-Umans modular composition.

The lifting of roots modulo $P$ to the message space takes time quadratic in $N$ with pre-processing requiring runtime exponent $\omega$ and storage quadratic in $N$. To this end, pre-compute a new basis $\{g_0,g_1,\ldots,g_{K-1}\}$ of $\mbox{Span}\{f_0,f_1,\ldots,f_{K-1}\}$. The basis $\{g_0,g_1,\ldots,g_{K-1}\}$ is chosen such that the $D$ by $K$ matrix $L$ over $\F_{q^2}$ whose $i^{th}$ column is $g_i(P)$ (written in a fixed basis for $\F_{q^{2D}}$ over $\F_{q^2}$) is lower triangular. Such a basis can be found by column reduction of the corresponding matrix whose $i^{th}$ column is $f_i(P)$, with runtime exponent $\omega$. In addition to $L$, we store the matrix $R$ expressing $\{g_0,g_1,\ldots,g_{K-1}\}$ as an $\F_{q^2}$-linear combination of  $\{f_0,f_1,\ldots,f_{K-1}\}$. Now given a residue modulo $P$, by solving the linear system corresponding to $L$, we can find an $\F_{q^2}$-linear combination of $\{g_0,g_1,\ldots,g_{K-1}\}$ that evaluates to that residue (if one exists). This takes quadratic time since $L$ is in lower triangular form. The matrix $R$ then expresses this lift as an $\F_{q^2}$-linear combination of $\{f_0,f_1,\ldots,f_{K-1}\}$, as desired.


\bibliography{general_bib}{}
\bibliographystyle{plain}

\end{document}